\documentclass[submission,copyright,creativecommons]{eptcs}
\providecommand{\event}{TARK 2023} 

\usepackage{iftex}

\ifpdf
  \usepackage{underscore}         
  \usepackage[T1]{fontenc}        
\else
  \usepackage{breakurl}           
\fi

\usepackage{amssymb}
\usepackage{amsmath}
\usepackage{stmaryrd}
\usepackage{amsthm}
\usepackage{comment}
\usepackage{url}
\usepackage{bussproofs}
\usepackage[all, knot]{xy  }
\usepackage{color}

\usepackage[english]{babel}
\usepackage{cancel}
\usepackage{comment}
\usepackage[ruled,vlined]{algorithm2e}
\usepackage{bm}
\usepackage{mathrsfs}

\newcommand{\BP}{\ensuremath{\mathbf{P}}}

\newcommand{\Act}{\ensuremath{\mathbf{A}}}
\newcommand{\Ag}{\ensuremath{\mathbf{I}}}

\newcommand{\M}{{M}}

\newcommand{\rel}[1]{\xrightarrow{#1}}

\newcommand{\lr}[1]{\langle #1 \rangle}

\newcommand{\ELKh}{\textsf{ELKh}}

\newcommand{\pspace}{P{\small SPACE}}

\newcommand{\Kh}{{K}h}
\newcommand{\K}{{K}}

\renewcommand{\phi}{\varphi}

\newcommand{\ACT}{\Act}

\newcommand{\Lan}{\ensuremath{\mathcal{L}}}

\newcommand{\CELeafN}{\ensuremath{\mathtt{ECE}}}

\newcommand{\Dom}{\ensuremath{\mathtt{dom}}}

\newcommand{\ruleNeg}{R\neg}
\newcommand{\ruleConj}{R\wedge}
\newcommand{\ruleDisj}{R\vee}
\newcommand{\ruleK}{R\K}
\newcommand{\ruleNegK}{R\neg{\K}}
\newcommand{\ruleT}{Cut\K}
\newcommand{\ruleFour}{R\K 4}
\newcommand{\ruleFive}{R\K 5}

\newcommand{\ruleKhtoKKh}{R\Kh 4}
\newcommand{\ruleKhFive}{R\Kh 5}

\newcommand{\newruleKh}{R{\Kh}}
\newcommand{\ruleNegKh}{R \neg\Kh}
\newcommand{\ruleCutKh}{Cut\ensuremath{\Kh}}
\newcommand{\ruleCutNegK}{Cut\ensuremath{\neg\K}}
\newcommand{\ruleCutNegKh}{Cut\ensuremath{\neg\Kh}}

\newcommand{\subpl}[1]{sub^+(#1)}
\renewcommand{\Vdash}{\vDash}

\newtheorem{theorem}{Theorem}
\newtheorem{definition}[theorem]{Definition}
\newtheorem{lemma}[theorem]{Lemma}
\newtheorem{proposition}[theorem]{Proposition}

\newcommand{\T}{\mathscr{T}}

\newcommand{\LeafNode}{\text{end node}}

\newcommand{\dep}[1]{\textit{dep}(#1)}

\title{Tableaux for the Logic of Strategically Knowing How}
\author{Yanjun Li
\institute{College of Philosophy,
Nankai University, 
Tianjin, China}
}
\def\titlerunning{Tableaux for the Logic of Strategically Knowing How}
\def\authorrunning{Yanjun Li}
\begin{document}
\maketitle

\begin{abstract}
	The logic of goal-directed \emph{knowing how} proposed in \cite{FervariEtAl2017} extends the standard epistemic logic 
	with an operator of \emph{knowing how}.
	The \emph{knowing how} operator is interpreted as that
	there exists a  strategy such that the agent knows that the strategy can make sure that $\phi$.
	This paper presents a tableau procedure for the multi-agent version of the logic of strategically \emph{knowing how} and shows the soundness and completeness of this tableau procedure.
	This paper also shows that the satisfiability problem of the logic can be decided in \pspace.
\end{abstract}

\section{Introduction}

Epistemic logic proposed by von Wright and Hintikka (see \cite{Wright51,Hintikka:kab}) is a logical formalism for reasoning about knowledge of agents.
It deals with propositional knowledge, that is, the knowledge expressed as \textit{knowing that} $\phi$ is true.
In recent years, other patterns of knowledge besides knowing that are attracting increasing attention in logic community, such as \textit{knowing whether} \cite{hart_knowing_1996,fan_wang_ditmarsch_2015}, \textit{knowing who} \cite{Epstein_Naumov_2021}, \textit{knowing the value} \cite{Baltag16,GuW16}, and \textit{knowing why} \cite{XuWS21} 
(see a survey in  \cite{WangBKT}).
Motivated by different scenarios in philosophy and AI, reasoning about \textit{knowing how} assertions are particularly interesting \cite{sep-knowledge-how}.

The discussion about formalizing the notion of \emph{knowing how} can date back to \cite{McCarthy79,moore1984formal}.
Currently, there are two main approaches of formalizing \emph{knowing how}.
One of them is connecting \emph{knowing how} with logics of \emph{knowing that} and \emph{ability} (see e.g. \cite{lesperance_ability_2000,andreasN2006}).
However, the main difficulty of this approach is that a simple combination of the modalities of \textit{knowing that} and \textit{ability} does not seem to capture a natural understanding of \emph{knowing how} (see a discussion in \cite{JamrogaA07,Herzig15}).
Instead of expressing \emph{knowing how} by \emph{knowing that} and \emph{ability} modalities, the other approach first adopted in \cite{Wang15lori} is expressing \emph{knowing how} in modal languages with a new modality of \emph{knowing how}.

Inspired by 
the idea of automated
planning under uncertainty in artificial intelligence,
The author of \cite{Wang15lori,Wang2018knowingHow} proposed a logical framework of \emph{knowing how} which includes a binary modality of \emph{knowing how}.
In \cite{Areces2021}, a new semantics for this \emph{knowing how} modality was given, which is based on an indistinguishability relation between plans. It is shown in \cite{Areces2021} that the satisfiability problem of this \emph{knowing how} logic is NP-complete, but this logic does not include the modality for  \emph{knowing that}.

Inspired by
the tradition of coalition logic, the authors of \cite{Naumov2017,Naumov2018a} introduced a logic to capture \emph{knowing how} of coalitions.
A coalition $C$ knows how to achieve $\phi$ if and only if there is a joint
action $a$ for $C$ such that it is distributed knowledge for $C$ that doing $a$ can make sure that $\phi$. 
Variants of the basic framework were proposed and discussed. 
In \cite{Naumov2018b}, a logic of \emph{knowing how} under the assumption of perfect recall was studied.
In \cite{NAUMOV201941}, a logic of \emph{knowing how} with the degree of uncertainty was discussed. 
In \cite{Naumov2018}, a logic of second-order \emph{knowing how} was proposed,  for the case that one coalition knows what the joint action of another coalition is. The topic of complexity is not covered in these literatures.

Along with the idea of formalizing \emph{knowing how} based on planning, the authors of \cite{FervariEtAl2017} proposed a single-agent logic of \emph{knowing how} via strategies. A strategy is a partial function from the set of agents' belief states into the set of actions.
The agent knows how to achieve $\phi$ if and only if there is a strategy such that all executions of the strategy will terminate on $\phi$-states.
Besides strategies, there are other types of plans, such as simple plans (i.e. a single action), linear plans (i.e. a sequence of actions), and so on.
The authors of \cite{liWang2021AIJ} proposed a unified logical framework to incorporate logics of \emph{knowing how} via different notions of plans.
They used a PDL-style programming language to syntactically specify various types of plans and discussed ten types of \emph{knowing how} based on ten different notions of plans.
It is shown in \cite{liWang2021AIJ} that the ten notions of plans lead to the same \emph{konwing how} logic, but over finite models,
the \emph{konwing how} logic based on knowledge-based plans requires an extra axiom, which leads to the same logic as the logic of strategically \emph{knowing how}. 

In \cite{Li21tableau}, a tableau-based decision procedure was proposed for the logic of \emph{knowing how} via simple plans.
This paper develops the method and presents a tableau procedure for the \emph{knowing how} logic via strategies proposed in \cite{FervariEtAl2017}. Strategically \emph{knowing how} can not be handled by the original method, since strategies are much more complicated than simple plans.
This paper also shows that the satisfiability problem of the logic of strategically \emph{knowing how} is in \pspace.  With other known results, this leads to the result that the satisfiability problem of the logic of strategically \emph{knowing how} is \pspace-complete.

The structure of this paper is as follows: Section 2 recalls the logic of strategically \emph{knowing how}; Section 3 presents a tableau procedure for the logic and proves its soundness; 
Section 4 shows the completeness of the tableau procedure and proves that the complexity of 
this logic is \pspace-complete. 
Section 5 concludes with some remarks.

\section{The logic of strategically knowing how}
This section presents the multi-agent version of the logic of strategically \emph{knowing how} from \cite{FervariEtAl2017}.

Let $\BP$ be a set of propositional letters and $\Ag$ be a set of agents where $|\Ag|\geq 2$.
\begin{definition}[$\ELKh_n$ Language]
	The Epistemic Language $\Lan_{\ELKh_n}$ of \emph{Knowing how} is defined by the following BNF where $p\in \BP$ and $i\in\Ag$:
	$$\phi::= \bot\mid p\mid \neg\phi\mid (\phi\land\phi)\mid \K_i\phi\mid \Kh_i{\phi}. $$
\end{definition}

We use $\top, \lor, \to$ as usual abbreviations. 
The formula $\K_i\phi$ means that the agent $i$ knows that $\phi$, and the formula $\Kh_i\phi$ means that the agent knows how to achieve the goal that $\phi$.

\begin{definition}[$\ELKh_n$ Models]
	A model $\M$ is a quintuple $\lr{W,\{\sim_i\mid i\in\Ag \},\{\ACT_i\mid i\in\Ag\},\{R_a \mid a\in \ACT_i,i\in\Ag \},V}$ where:
	\begin{itemize}
		\item $W$ is a non-empty set of states,
		\item $\sim_i\ \subseteq {W\times W}$ is an equivalence relation for each $i\in\Ag$,
		\item $\ACT_i$\ is a set of actions for each $i\in\Ag$,
		\item $R_a  \subseteq {W\times W}$ is a binary relation for each 
		      $a\in \ACT$ where $\ACT=\bigcup_{i\in\Ag}\ACT_i$,
		\item $V:W\to 2^\BP$ is a valuation function.
	\end{itemize}
\end{definition}

Given $s\in W$, we use $[s]^i$ to denote the equivalence class of $s$ over $\sim_i$, i.e., $[s]^i=\{t\in W\mid s\sim_i t\}$, and use $[W]^i$ to denote the set of all equivalence classes of states in $W$ over $\sim_i$, namely $[W]^i=\{[s]^i\mid s\in W\}$.
We say that the action $a\in\Act_i$ is \emph{executable} at $s$ if $(s,t)\in R_a$ for some $t$.
We use $[s]^i\rel{a}[t]^i$ to denote that there are some $s'\in [s]^i$ and some $t'\in [t]^i$ such that $(s',t')\in R_a$.


\begin{definition}[Strategies]
	Given a model $\M$, a \emph{uniformly executable strategy} (or simply called \emph{strategy}) for agent $i$ in $\M$ is a partial function $\sigma:[W]^i\to \ACT_i$ such that $\sigma([s]^i)$ is executable at all $s'\in[s]^i$. Particularly, the empty function is also a strategy, the \emph{empty strategy}.
\end{definition}

We use $\Dom(\sigma)$ to denote the domain of $\sigma$.

\begin{definition}[Executions]
	Given a strategy $\sigma$ of agent $i$ in $\M$, a \emph{possible execution} of $\sigma$ is a possibly infinite sequence of equivalence classes $\delta=[s_0]^i[s_1]^i\cdots$ such that $[s_{k}]^i\rel{\sigma([s_{k}]^i)}[s_{k+1}]^i$ for all $0\leq k< |\delta|$.
	Particularly, $[s]^i$ is a possible execution if $[s]^i\not\in\Dom(\sigma)$.
	If the execution $\rho$ is a finite sequence $[s_0]^i\cdots[s_k]^i$, we call $[s_k]^i$ the \LeafNode\ of $\rho$.
	A possible execution  of $\sigma$ is \emph{complete} if it is infinite or its \LeafNode\ is not in $\Dom(\sigma)$.
\end{definition}

Given an $i$-strategy $\sigma$, we use $\CELeafN(\sigma,[s]^i)$ to denote the set of all  \LeafNode s of all $\sigma$'s complete executions
starting from $[s]^i$.


\begin{definition}[$\ELKh_n$ Semantics] 
	The satisfaction relation $\vDash$ between a pointed model ($\M,s$) and a formula $\phi$ is defined 
	as follows:
	\[\begin{array}{lll}
			\M,s\vDash p             & \iff & s\in V(p)                                                                                              \\
			\M,s\vDash \neg\phi      & \iff & \M,s\nvDash \phi                                                                                       \\
			\M,s\vDash \phi\land\psi & \iff & \M,s\vDash\phi \text{ and } \M,s\vDash \psi                                                            \\
			\M,s\vDash \K_i{\phi}    & \iff & \text{for all }s': \text{ if }s{\sim_i}s'\text{ then }\M,s'\vDash\phi                                  \\
			\M,s\vDash \Kh_i{\phi}   & \iff & \text{there exists a strategy $\sigma$ for agent $i$ such that }                                       \\
			                         &      & 1.\text{ all $\sigma$'s complete executions  starting from $[s]^i$ are finite, and }                   \\ 
			                         &      & 2. [t]^i{\subseteq}\{s'\in W\mid \M,s'\vDash\phi \} \text{ for all }[t]^i {\in}\CELeafN(\sigma,[s]^i).
			\\
		\end{array}\]
\end{definition}


\begin{proposition}\label{pro.validFormulas}
	The following formulas are valid.


	\begin{enumerate}
		\item[(1).] $\K_i\phi\to\Kh_i\phi$
		\item[(2).] $\Kh_i\phi\to \K_i\Kh_i\phi$
		\item[(3).] $\neg\Kh_i\phi\to \K_i\neg\Kh_i\phi$
		\item[(4).] $\Kh_i\phi\to\Kh_i\K_i\phi$
		\item[(5).] $\Kh_i\Kh_i\phi\to\Kh_i\phi$
	\end{enumerate}
\end{proposition}
\begin{proof}
	For (1), it is due to the empty strategy. For (2), (3), and (4), it follows from the semantics. For (5), see~\cite{FervariEtAl2017}.
\end{proof}


\section{Tableaux}
This section presents a tableau procedure for the logic $\ELKh_n$ and shows the soundness of the tableau procedure.

Given $\phi$, let $\subpl{\phi}$ be the set $\{\psi,\neg\psi \mid \psi$ is a subformula of $\phi\}\cup\{\K_i\psi,\neg\K_i\psi\mid \Kh_i\psi$ is a subformula of $\phi\}$.

\subsection{Tableau procedure}

A tableau is a rooted 
tree in which each node is labeled with a set of prefixed formulas.
A prefixed formula is a pair $\lr{\sigma,\phi}$ where the prefix $\sigma$ is an alternative sequence of natural numbers and agents or $\Kh$-formulas, such as $\lr{1 i 2 {\Kh_i p} 3, \phi}$.
The prefixes represent states in models.
The agent symbols and $\Kh$-formulas occurring in $\sigma$ indicate some epistemic information and action information on the current state.
For example, the prefix $1 i 2 {\Kh_i p} 3$ indicates the following informations: there are three states $1$, $2$, and $3$; $1\sim_i 2$; there is an action $a$ such that $a$ is a good plan for $2\vDash\Kh_i p$ and $(2,3)\in R_a$.

\begin{definition}[Tableaux]
	A tableau for $\phi_0$ is a labeled tree that is defined as follows:
	\begin{itemize}
		\item[A.] Create the root node and label it with $\lr{0,\phi_0}$;
		\item[B.]  Extend the tree by rules in Table \ref{tab:tableau}.
	\end{itemize}
\end{definition}

\begin{table}[htbp]
	\centering
	\begin{tabular}{ll}
		\AxiomC{$\lr{\sigma,\neg\neg\phi}$}
		\LeftLabel{$(\ruleNeg)$}
		\UnaryInfC{$\lr{\sigma,\phi}$}
		\DisplayProof
		\vspace*{1em}
		 &                                            \\

		\AxiomC{$\lr{\sigma,\neg(\phi_1\wedge\phi_2)}$}
		\LeftLabel{$(\ruleDisj)$}
		\UnaryInfC{
			\begin{tabular}{l|l|l}
				$\lr{\sigma,\neg\phi_1}$ & $\lr{\sigma,\neg\phi_1}$ & $\lr{\sigma,\phi_1}$     \\
				$\lr{\sigma,\neg\phi_2}$ & $\lr{\sigma,\phi_2}$     & $\lr{\sigma,\neg\phi_2}$ \\
			\end{tabular}
		}
		\DisplayProof
		 &
		\AxiomC{$\lr{\sigma,\phi_1\wedge\phi_2}$}
		\LeftLabel{$(\ruleConj)$}
		\UnaryInfC{\parbox{40pt}{$\lr{\sigma,\phi_1}$ \\$\lr{\sigma,\phi_2}$}}
		\DisplayProof
		\vspace*{1em}

		\\

		\AxiomC{$\lr{\sigma,\neg\K_i\phi}$}
		\LeftLabel{(\ruleCutNegK)}
		\RightLabel{}
		\UnaryInfC{
			$\lr{\sigma,\neg\phi}\mid \lr{\sigma,\phi}$
		}
		\DisplayProof
		\vspace*{1em}
		 &

		\AxiomC{$\lr{\sigma,\K_i\phi}$}
		\LeftLabel{$(\ruleT)$}
		\UnaryInfC{$\lr{\sigma,\phi}$}
		\DisplayProof
		\vspace*{1em}
		\\

		\AxiomC{$\lr{\sigma,\Kh_i\phi}$}
		\LeftLabel{(\ruleCutKh)}
		\RightLabel{}
		\UnaryInfC{
			\begin{tabular}{c|c}
				$\lr{\sigma,\neg\K_i\phi}$ & $\lr{\sigma,\K_i\phi}$ \\
			\end{tabular}
		}
		\DisplayProof
		\vspace*{1em}
		 &
		\AxiomC{$\lr{\sigma,\neg\Kh_i\phi}$}
		\LeftLabel{(\ruleCutNegKh)}
		\UnaryInfC{$\lr{\sigma,\neg\K_i\phi}$}
		\DisplayProof
		\\ 
		\AxiomC{
			\begin{tabular}{c}
				$\lr{\sigma,\neg\K_i\phi}$ \\
			\end{tabular}
		}
		\LeftLabel{$(\ruleNegK)$}
		\RightLabel{$n'$ is new}
		\UnaryInfC{$\lr{\sigma i n',\neg\phi}$
		}
		\DisplayProof
		\vspace*{1em}
		 &
		\AxiomC{$\lr{\sigma,\K_i\phi}$}
		\LeftLabel{$(\ruleK)$}
		\RightLabel{$\sigma i n'$ is used}
		\UnaryInfC{$\lr{\sigma i n',\phi}$}
		\DisplayProof
		\vspace*{1em}
		\\

		\AxiomC{$\lr{\sigma,\K_i\phi}$}
		\LeftLabel{$(\ruleFour)$}
		\RightLabel{$\sigma i n'$ is used}
		\UnaryInfC{$\lr{\sigma i n',\K_i\phi}$}
		\DisplayProof
		\vspace*{1em}
		 &
		\AxiomC{$\lr{\sigma,\neg\K_i\phi}$}
		\LeftLabel{$(\ruleFive)$}
		\RightLabel{$\sigma i n'$ is used}
		\UnaryInfC{$\lr{\sigma i n',\neg\K_i\phi}$}
		\DisplayProof
		\vspace*{1em}
		\\

		\AxiomC{$\lr{\sigma,\Kh_i\phi}$}
		\LeftLabel{$(\ruleKhtoKKh)$}
		\RightLabel{$\sigma i n'$ is used}
		\UnaryInfC{$\lr{\sigma i n',\Kh_i\phi}$}
		\DisplayProof
		\vspace*{1em}
		 &
		\AxiomC{$\lr{\sigma,\neg\Kh_i\phi}$}
		\LeftLabel{$(\ruleKhFive)$}
		\RightLabel{$\sigma i n'$ is used}
		\UnaryInfC{$\lr{\sigma i n',\neg\Kh_i\phi}$}
		\DisplayProof
		\\

		\AxiomC{
			\begin{tabular}{l}
				$\lr{\sigma,\Kh_i\phi}$ \\
				$\lr{\sigma,\neg\K_i\phi}$
			\end{tabular}
		}
		\LeftLabel{$(\newruleKh)$}
		\RightLabel{$n'$ is new}
		\UnaryInfC{
			\begin{tabular}{l}
				$\lr{\sigma \Kh_i\phi n',\K_i\phi}$
			\end{tabular}
		}
		\DisplayProof
		\vspace*{1em}

		 &

		\AxiomC{
			\begin{tabular}{l}
				$\lr{\sigma,\neg\Kh_i\phi}$ \\
				$\lr{\sigma,\Kh_i\psi}$     \\
				$\lr{\sigma,\neg\K_i\psi}$
			\end{tabular}
		}
		\LeftLabel{$(\ruleNegKh)$}
		\RightLabel{$n'$ is new}
		\UnaryInfC{
			\begin{tabular}{l}
				$\lr{\sigma \Kh_i\psi n',\K_i\psi}$      \\
				$\lr{\sigma \Kh_i\psi n',\neg\Kh_i\phi}$ \\
			\end{tabular}
		}
		\DisplayProof
		\\
	\end{tabular}
	\caption{Tableau rules}\label{tab:tableau}
\end{table}

Next, we will give a procedure to construct a tableau (Definition \ref{def.tableaux} below). In Section 3.2, we will show that the procedure is sound, and we in Section 4.1 will show that it is complete, and we in Section 4.2 will show that it runs in polynomial space.
Before that, we first introduce some auxiliary notations below.

Let $\Act$ be a set of actions. We use $\Act^+$ to denote the set $\Act\cup\{\epsilon\}$.
Let $\Gamma$ be a set of formulas. We use $\Gamma|\K_i$, $\Gamma|\neg\K_i$ and $\Gamma|\Kh_i$ to respectively denote the set $\{\phi\in\Gamma \mid \phi$ is of the form $\K_i\psi \}$, $\{\phi\in\Gamma\mid \phi$ is of the form $\neg\K_i\psi \}$ and $\{\phi\in\Gamma\mid \phi$ is of the form $\Kh_i\psi \}$.
We say that a formula set $\Gamma$ is \emph{blatantly inconsistent} iff either $\phi,\neg\phi\in \Gamma$ for some formula $\phi$ or $\bot\in\Gamma$.

A labeled tree $\T$ is a triple $\lr{N,E,L}$,
where $\lr{N,E}$ is a rooted tree with the node set $N$ and the edge set $E$
and $L:N\cup E\to \mathcal{P}(\Lan_{\ELKh_n})\cup\Ag\cup\Act^+$ is a label function such that each node is labeled a formula set and each edge is labeled an agent $i\in\Ag$ or an action $a\in \Act^+$.
A node sequence $n_1 \cdots n_{h+1}$ is a \emph{path} in $\T$ if $(n_k,n_{k+1})\in E$ for all $1\leq k\leq h$.
If the node sequence $n_1 \cdots n_{h+1}$ is a path in $\T$ and the label $L(n_k,n_{k+1})$ is either $i\in\Ag$ or $\epsilon$ for all $1\leq k\leq h$, we then say that $n_1$ is an \emph{$i$-ancestor} of $n_{h+1}$.

A tree is called an \emph{and-or} tree if each non-leaf node is marked as ``and" node or ``or" node.
A subtree of an and-or tree is called \emph{complete} if it contains the root node, and each ``and" non-leaf node has all its child nodes, and each ``or" non-leaf node has at least one child node.

Now we are ready to give a procedure to construct a tableau.
We remark that, strictly speaking, the tree constructed by the following procedure is not really a tableau.
Rather, it is a tree in which the desired tableau is embedded. Such trees are called \emph{pre-tableaux} in \cite{halpern_guide_1992}.
Since in the remaining paper, we will work only on the following procedure and show that the following procedure is sound and complete and runs in polynomial space, it does not matter what we call it. So, in the remaining paper, we will call the tree constructed by the following procedure a tableau.

\begin{definition}[Tableaux construction]\label{def.tableaux}
	A tableau for $\phi_0$ is a labeled and-or tree $\T_{\phi_0}$
	which is constructed by the following steps:

	\begin{itemize}
		\item[(I).] Construct a tree consisting of a single node $n_0$ (i.e. the root node), and label the root node the formula set $\{\phi_0\}$. 
		\item[(II).] Repeat until none of (1)-(2) below applies:
			\begin{itemize}
				\item[(1).] \textit{Forming a subformula-closed propositional tableau}: if $n$ is an unblocked leaf node and $L(n)$ is not blatantly inconsistent,
					then mark $n$ as an ``or" node and check the first unchecked formula $\phi\in L(n)$ at $n$ by the following:
					\begin{itemize}
						\item[(a).] 
							If $\phi$ is of the form $\neg\neg\psi $ and $\psi$ is not in $L(n)$, then create a successor node $n'$ of $n$, set

							\[
								\begin{array}{ll}
									 & L(n')=L(n)\cup\{\psi\}, \\
									 & L(n,n')=\epsilon,
								\end{array}
							\]

							and mark $\phi$ and all checked formulas at $n$ as ``checked" at $n'$.

						\item[(b).] 
							If $\phi$ is of the form $\phi_1\land\phi_2 $ and either $\phi_1$ or $\phi_2$ is not in $L(n)$, then create a successor node $n'$ of $n$, set
							\[
								\begin{array}{ll}
									 & L(n')=L(n)\cup\{\phi_1,\phi_2\}, \\
									 & L(n,n')=\epsilon,
								\end{array}
							\]

							and mark $\phi$ and all checked formulas at $n$ as ``checked" at $n'$

						\item[(c).] 
							If $\phi$ is of the form $\neg(\phi_1\land\phi_2)$ and none of the three sets $\{\neg\phi_1,\neg\phi_2\}$, $\{\neg\phi_1,\neg\phi_2\}$, $\{\neg\phi_1,\neg\phi_2\}$ is a subset of $L(n)$, then
							create three successors $n_1,n_2,n_3$ of $n$, set

							\[
								\begin{array}{ll}
									 & L(n_1)    =  L(n)\cup\{\neg\phi_1,\neg\phi_2\}, \\
									 & L(n_2)    =  L(n)\cup\{\neg\phi_1,\phi_2\},     \\
									 & L(n_3)    =  L(n)\cup\{\phi_1,\neg\phi_2\},     \\
									 & L(n,n_1)  =  L(n,n_2)=L(n,n_3)=\epsilon,
								\end{array}
							\]

							and mark $\phi$ and all checked formulas at $n$ as ``checked" at $n_1,n_2,n_3$

						\item[(d).] 
							If $\phi$ is of the form $\K_i\psi$ and $\psi$ is not in $L(n)$, then create a successor node $n'$ of $n$, and set
							\[
								\begin{array}{ll}
									 & L(n')=L(n)\cup\{\psi\}, \\
									 & L(n,n')=\epsilon,
								\end{array}
							\]

							and mark $\phi$ and all checked formulas at $n$ as ``checked" at $n'$

						\item[(e).] 
							If $\phi$ is of the form $\neg\K_i\psi$ and neither $\neg\psi$ nor $\psi$ is in $L(n)$, then create two successors $n_1,n_2$ of $n$, set
							\[
								\begin{array}{ll}
									 & L(n_1)=L(n)\cup\{\neg\psi\}, \\
									 & L(n_2)=L(n)\cup\{\psi\},     \\
									 & L(n,n_1)=L(n,n_2)=\epsilon
								\end{array}
							\]

							and mark $\phi$ and all checked formulas at $n$ as ``checked" at $n_1$ and $n_2$

						\item[(f).] 
							If $\phi$ is of the form $\Kh_i\psi$ and neither $\neg\K_i\psi$ nor $\K_i\psi$ is in $L(n)$, then create two successors $n_1,n_2$ of $n$, set
							\[
								\begin{array}{ll}
									 & L(n_1)=L(n)\cup\{\neg\K_i\psi\}, \\
									 & L(n_2)=L(n)\cup\{\K_i\psi\},     \\
									 & L(n,n_1)=L(n,n_2)=\epsilon,
								\end{array}
							\]

							and mark $\phi$ and all checked formulas at $n$ as ``checked" at $n_1$ and $n_2$.

						\item[(g).] 
							If $\phi$ is of the form $\neg\Kh_i\psi$ and $\neg\K_i\psi$  is not in $L(n)$, then create a successor $n'$ of $n$, set
							\[
								\begin{array}{ll}
									 & L(n')=L(n)\cup\{\neg\K_i\psi\}, \\
									 & L(n,n')=\epsilon,
								\end{array}
							\]

							and mark $\phi$ and all checked formulas at $n$ as ``checked" at $n'$

						\item[(h).] if none of (a)-(g) above applies, then mark $\phi$ as ``checked".
					\end{itemize}
				\item[(2).]
					\textit{Creating successors for $\neg\K_i,\Kh_i,\neg\Kh_i$ formulas}: if $n$ is an unblocked leaf node, $L(n)$ is not blatantly inconsistent, and each formula in $L(n)$ is marked as ``checked", then mark (or re-mark) $n$ as an ``and" node and do the following:
					\begin{itemize}
						\item[(i).] For each formula in $L(n)$ of the form $\neg\K_i\psi$ , if there is no $i$-ancestor of $n'$ such that $L(n')=\Sigma(n,\neg\K_i\psi)$, then create a successor $n_{\neg\K_i\psi}$ of $n$ and set
							\[
								\begin{array}{ll}
									 & L(n_{\neg\K_i\psi})=\Sigma(n,\neg\K_i\psi), \\
									 & L(n,n_{\neg\K_i\psi})=i
								\end{array}
							\]
							in which
								{$$\Sigma(n,\neg\K_i\psi)=\{\neg\psi\}\cup (L(n)|\K_i) \cup (L(n)|\neg\K_i)\cup (L(n)|\Kh_i)\cup (L(n)|\neg\Kh_i).$$}


						\item[(j).] For each pair on $L(n)$ of the form $(\Kh_i\psi,\neg\K_i\psi)$, create a successor $n_{\Kh_i\psi}$ of $n$, and set
							\[
								\begin{array}{ll}
									 & L(n_{\Kh_i\psi})=\{\K_i\psi\},    \\
									 & L(n,n_{\Kh_i\psi})=a_{\Kh_i\psi}.
								\end{array}
							\]


						\item[(k).] For each triple on $L(n)$ of the form $(\neg\Kh_i\chi,\Kh_i\psi,\neg\K_i\psi)$, create a successor $n_{(\neg\Kh_i\chi,\Kh_i\psi)}$, and set
							\[
								\begin{array}{ll}
									 & L(n_{(\neg\Kh_i\chi,\Kh_i\psi)})=\{\K_i\psi,\neg\Kh_i\chi\}, \\
									 & L(n,n_{(\neg\Kh_i\chi,\Kh_i\psi)})=a_{\Kh_i\psi}.
								\end{array}
							\]

							Moreover, if there is an ancestor $n'$ of $n_{(\neg\Kh_i\chi,\Kh_i\psi)}$ such that $L(n')=L(n_{(\neg\Kh_i\chi,\Kh_i\psi)})$, then mark $n_{(\neg\Kh_i\chi,\Kh_i\psi)}$ as blocked, and we say that $n_{(\neg\Kh_i\chi,\Kh_i\psi)}$ is blocked by $n'$.

					\end{itemize}
			\end{itemize}
	\end{itemize}
\end{definition}

\begin{definition}
	A subtree of a tableau is \emph{closed} if there is some node $n$ in it such that $L(n)$ is blatantly inconsistent. Otherwise, it is called \emph{open}. A tableau is \emph{closed} iff all its complete subtrees are closed. 
\end{definition}

\subsection{Soundness}
In this subsection, we will show that the procedure of Definition \ref{def.tableaux} is sound.

\begin{definition}[Interpretations]
	Given a model $\M$ 
	and a subtree $\T'=\lr{N',E',L'}$ of  the tableau $\T_{\phi_0}$, 
	let $f$ be a function from $N'$ to $W$.
	We say that $f$ is an \emph{interpretation} of $\T'$ if and only if $\M,f(n)\Vdash\phi$ for all $\phi\in L(n)$ and all $n\in N'$.
\end{definition}

\begin{lemma}\label{lemma.forSound}
	If $\M,s\vDash\phi_0$, then there exists an interpretation of some  complete subtree of any $\T_{\phi_0}$.
\end{lemma}
\begin{proof}
	Let $f_0$ be the function $f_0=\{n_0\mapsto  s\}$ which maps the root $n_0$ of $\T_{\phi_0}$ to the state $s$. It is obvious that $\M,f_0(n_0)\vDash \phi$ for all $\phi\in L(n_0)$.

	Then we only need to show the following statement:

	\begin{center}
		\parbox{0.8\textwidth}{
			If $n$ is a leaf node of a subtree $\T$, $f$ is an interpretation of $\T$, and   $n$ is in the domain of $f$, then\\
			(A): for each construction steps (a)-(g), there is an interpretation of $\T$ extending with one child node of $n$, and \\
			(B): for each construction steps (i)-(k), there is an interpretation of $\T$ extending with all child nodes of $n$.
		}
	\end{center}


	For (A), firstly it is obvious for the steps (a)-(c).
	For the step (d), it follows that $\sim_i$ is a reflexive relation.
	For the steps (e) and (f), it follows from the fact that a formula is either true or false on a state.
	For the step (g), it follows from the fact that $\neg\Kh_i\psi\to \neg\K_i\psi$ is valid (see Proposition \ref{pro.validFormulas}).

	Next, we will show that (B) holds.

	For the step (i), for each $n$'s child node $n_{\neg\K_i\psi}$, we know that $\neg\K_i\psi\in L(n)$ and $L(n_{\neg\K_i\psi})=\{\neg\psi\}\cup (L(n)|\K_i) \cup (L(n)|\neg\K_i)\cup (L(n)|\Kh_i)\cup (L(n)|\neg\Kh_i)$.
	Since $f$ is an interpretation of $\T$ including $n$, it follows that $\M,f(n)\vDash\neg\K_i\psi$. Hence, there exists a state $t_{\neg\K_i\psi}\in [f(n)]^i$ such that $\M,t_{\neg\K_i\psi}\vDash\neg\psi$. Moreover, since $\sim_i$ is an equivalence relation, $(\M,t_{\neg\K_i\psi})$ satisfies all the $\K_i$-formulas and $\neg\K_i$-formulas that are true at $f(n)$.
	Furthermore, by Proposition \ref{pro.validFormulas}, we have that $(\M,t_{\neg\K_i\psi})$ satisfies all the $\Kh_i$-formulas and $\neg\Kh_i$-formulas that are true at $f(n)$.
	Let $f'$ be the $f$-extension $f\cup\{n_{\neg\K_i\psi}\mapsto t_{\neg\K_i\psi}\mid n_{\neg\K_i\psi}$ is a child node of $n \}$.  Therefore, we have that $\M,f'(n_{\neg\K_i\psi})\vDash L(n_{\neg\K_i\psi})$.

	For the step (j), for each $n$'s child node $n_{\Kh_i\psi}$, we know that $\Kh_i\psi\in L(n)$ and $L(n_{\neg\K_i\psi})=\{\K_i\psi \}$.
	Since $f$ is an interpretation of $\T$ including $n$, it follows that $\M,f(n)\vDash\Kh_i\psi$.
	So, by the semantics, there exist an $i$-strategy $\sigma$ and a state $t_{\Kh_i\psi}\in \CELeafN(\sigma,[f(n)]^i)$ such that $\M,t_{\Kh_i\psi}\vDash\K_i\psi$.
	Let $f'$ be the $f$-extension $f\cup\{n_{\Kh_i\psi}\mapsto t_{\Kh_i\psi}\mid n_{\Kh_i\psi}$ is a child node of $n \}$.  Therefore, we have that $\M,f'(n_{\Kh_i\psi})\vDash L(n_{\Kh_i\psi})$.

	For the step (k), for each $n$'s child node $n_{(\neg\Kh_i\chi,\Kh_i\psi)}$, we know that $\neg\Kh_i\chi,\Kh_i\psi\in L(n)$ and $L(n_{(\neg\Kh_i\chi,\Kh_i\psi)})=\{\K_i\psi,\neg\Kh_i\chi \}$.
	Since $f$ is an interpretation of $\T$ including $n$, it follows that $\M,f(n)\vDash\neg\Kh_i\chi\land\Kh_i\psi$.
	Due to $\M,f(n)\vDash\Kh_i\psi$, it follows by the semantics that there exist an $i$-strategy $\sigma$ such that $\M,t\vDash\K_i\psi$ for all $[t]^i\in \CELeafN(\sigma,[f(n)]^i)$.
	Moreover, it must be the case that there exists $t_{(\neg\Kh_i\chi,\Kh_i\psi)}$ such that $[t_{(\neg\Kh_i\chi,\Kh_i\psi)}]^i\in\CELeafN(\sigma,[f(n)]^i)$ and $\M,t_{(\neg\Kh_i\chi,\Kh_i\psi)}\vDash\neg\Kh_i\chi$.
	Otherwise, if $\M,t\vDash\Kh_i\chi$ for all $[t]^i\in\CELeafN(\sigma,[f(n)]^i)$, this implies $\M,f(n)\vDash\Kh\Kh_i\chi$.
	By Proposition \ref{pro.validFormulas}, it follows that $\M,f(n)\vDash\Kh_i\chi$, which is contradictory with the fact that $\M,f(n)\vDash\neg\Kh_i\chi$.
	Hence,  there exists $t_{(\neg\Kh_i\chi,\Kh_i\psi)}$ such that $[t_{(\neg\Kh_i\chi,\Kh_i\psi)}]^i\in\CELeafN(\sigma,[f(n)]^i)$ and $\M,t_{(\neg\Kh_i\chi,\Kh_i\psi)}\vDash\neg\Kh_i\chi$, which implies that $\M,t_{(\neg\Kh_i\chi,\Kh_i\psi)}\vDash\neg\Kh_i\chi\land\K_i\psi$.
	Let $f'$ be the $f$-extension $f\cup\{n_{(\neg\Kh_i\chi,\Kh_i\psi)}\mapsto t_{(\neg\Kh_i\chi,\Kh_i\psi)}\mid n_{(\neg\Kh_i\chi,\Kh_i\psi)}$ is a child node of $n \}$.  Therefore, we have that $\M,f'(n_{(\neg\Kh_i\chi,\Kh_i\psi)})\vDash L(n_{(\neg\Kh_i\chi,\Kh_i\psi)})$.
\end{proof}

The soundness below follows from Lemma \ref{lemma.forSound} above.

\begin{theorem}[Soundness]\label{theo.sound}
	If $\phi_0$ is satisfiable, then $\T_{\phi_0}$ is not closed. 
\end{theorem}

\section{Completeness and complexity}


\subsection{Completeness}
In this subsection, we will show that the procedure of Definition \ref{def.tableaux} is complete.

Recall that $\subpl{\phi}$ is the set $\{\psi,\neg\psi \mid \psi$ is a subformula of $\phi\}\cup\{\K_i\psi,\neg\K_i\psi\mid \Kh_i\psi$ is a subformula of $\phi\}$.
From the construction of $\T_{\phi_0}$, it follows that $L(n)\subseteq \subpl{\phi_0}$ for each node $n$.
Moreover, each formula in $L(n)$ is ``inherited" from $n$'s ancestors, that is, if $\psi\in L(n)$ and $n'$ is an ancestor of $n$ then there exists $\phi\in L(n')$ such that $\psi\in \subpl{\phi}$.

A path $n_1\cdots n_h$ of $\T_{\phi_0}$ is called an $\epsilon$-path iff $L(n_k,n_{k+1})=\epsilon$ for all $1\leq k<h$.
Particularly, a path with length 1 is an $\epsilon$-path.
An $\epsilon$-path $n_1\cdots n_h$ is \emph{maximal} iff there are no such nodes $n$ and $n'$ that either $nn_1\cdots n_h$ or $n_1\cdots n_h n'$ is an $\epsilon$-path.
Given a path $\rho=n_1\cdots n_h$, we use $ini(\rho)$, $end(\rho)$ and $L(\rho)$ to respectively denote $n_1$, $n_h$ and $L(n_h)$.
We use $\rho\rel{x}\rho'$ to means that $\rho\rho'$ is a path and $L(end(\rho),ini(\rho'))=x$.

From the construction of $\T_{\phi_0}$ 
we know that if a node $n$ is blocked then there is a unique node $n'$ which blocks $n$ and $n$ itself is a maximal $\epsilon$-path.
Moreover, if $\rho$ is a maximal $\epsilon$-path and $end(\rho)$ is blocked, then $\rho$ is a single node, i.e. $\rho=n$ where the node $n$ is blocked.

For each maximal $\epsilon$-path $\rho$ of $\T_{\phi_0}$,
by the construction, we know that if $\rho$ is not blocked and $L(\rho)$ is not blatantly inconsistent then each formula in $L(\rho)$ is marked as ``checked".
This means that $L(\rho)$ is closed over $sub^+$-formulas, that is,
if $\phi\in L(\rho)$ then either $\psi$ or $\sim\psi$ is in $L(\rho)$ for all $\psi\in\subpl{\phi}$, where $\sim\psi=\chi$ if $\psi=\neg\chi$, otherwise $\sim\psi=\neg\psi$.


\begin{definition}[$\M^\T$]\label{def.inducedModel}
	Let $\T$ be a complete subtree of $\T_{\phi_0}$. The model induced by $\T$, denoted by $\M^\T$, is defined as follows:
	\begin{itemize}
		\item $W=\{\rho\mid \rho$ is a maximal $\epsilon$-path of $\T$, and $\rho$ is not blocked$\}$,
		\item $\sim_i$ is the minimal equivalence relation $X$ on $W$ such that $\{(\rho,\rho')\mid \rho\rel{i} \rho' \}\subseteq X$, 
		\item $\Act_i=\{a_{\Kh_i\psi} 
			      \mid$ there exists an edge $(n,n')$ in $\T$ such that $L(n,n')=a_{\Kh_i\psi}\}$. 
		\item for each $a\in \Act_i$, $R_a=\{(\rho,\rho')\mid \rho\rel{a} \rho'  $, or $ \rho\rel{a}\rho''$ where $\rho''$ is a maximal $\epsilon$-path blocked by $\rho' \}$.
		\item $f(\rho)\in V(p)$ iff $p\in L(f(\rho))$.
	\end{itemize}
\end{definition}

Please note that by the definition $\Act_i$, we know that if $i\neq j$ then $\Act_i\cap\Act_j=\emptyset$.

The following proposition states that all paths in the same $\sim_i$-closure share the same formulas in the forms $\K_i\phi$, $\neg\K_i\phi$, $\Kh_i\phi$, or $\neg\Kh_i\phi$.

\begin{proposition}\label{pro.equivalenceShareKAndKh}
	Let $\T$ be a complete and open subtree of $\T_{\phi_0}$.
	If $\rho_1\sim_i\rho_2$ in $\M^\T$, then $L(\rho_1)|x = L(\rho_2)|x$ where $x\in\{\K_i,\neg\K_i,\Kh_i,\neg\Kh_i \}$. 
\end{proposition}
\begin{proof}
	Besides $\rho_1=\rho_2$, there are three possible cases: (1) $end(\rho_1)$ is an $i$-ancestor of $end(\rho_2)$, or (2) $end(\rho_2)$ is an $i$-ancestor of $end(\rho_1)$, or (3) there is some $\rho\in W$ such that $end(\rho)$ is an $i$-ancestor of both $end(\rho_1)$ and $end(\rho_2)$.

	For (1), firstly we will show that $L(\rho_1)|x \subseteq L(\rho_2)|x$.  In the construction step (i) all $x$-formulas are inherited from the $i$-ancestor. Moreover, in the steps (a)-(f) all formulas labeled on the parent node are inherited by each child node. Thus, $L(\rho_1)|x \subseteq L(\rho_2)|x$.

	Secondly, we will show that $L(\rho_2)|x\subseteq L(\rho_1)|x$. Since $end(\rho_1)$ is an ancestor of $end(\rho_2)$, we know from the construction that for each formula $\psi\in L(\rho_2)$ there exists $\phi\in L(\rho_1)$ such that $\psi\in\subpl{\phi}$.
	Moreover, since $\rho_1$ is a  non-blocked maximal $\epsilon$-path and $L(\rho_1)$ is not blatantly inconsistent, it follows that all formulas in $L(\rho_1)$  are marked as ``checked". So, $L(\rho_1)$ is closed over $sub^+$-formulas.
	Hence, for each $x$-formula $\chi\in L(\rho_2)$, either $\chi$ or $\sim\chi$ is in $L(\rho_1)$.
	If $\sim\chi\in L(\rho_1)$, since $\sim\chi$ will be inherited to $L(\rho_2)$, it means that $L(\rho_2)$ will be blatantly inconsistent, which is contradictory with the fact that $\T$ is open. Therefore, it only can be that $\chi\in L(\rho_1)$. 
	Hence, $L(\rho_2)|x\subseteq L(\rho_1)|x$.

	For (2), it can be shown by the same process as for (1).

	For (3), by (1) we know that $L(\rho)|x=L(\rho_1)|x$ and $L(\rho)|x=L(\rho_2)|x$. Thus, $L(\rho_1)|x=L(\rho_2)|x$.
\end{proof}


Before we show the truth lemma, we need the following auxiliary proposition.

\begin{proposition}\label{pro.negKhWillBeInheritedByActionTransitions}
	Let $\T$ be a complete subtree of $\T_{\phi_0}$, $\sigma$ be an $i$-strategy in $\M^\T$, and $\delta=[\rho_1]^i\cdots [\rho_h]^i$ be a $\sigma$-execution in $\M^\T$. If $\neg\Kh_i\chi\in L(\rho_h)$ and $\delta$ is not complete, then there exists some $\rho\in W$ such that $\delta [\rho]^i$is a $\sigma$-execution and $\neg\Kh_i\chi\in L(\rho)$.
\end{proposition}
\begin{proof}
	Since $[\rho_1]^i\cdots [\rho_h]^i$ is not complete, it follows that $[\rho_h]^i\in \Dom(\sigma)$. Let $\sigma([\rho_h]^i)=a_{\Kh_i\psi}$ for some $a_{\Kh_i\psi}\in \Act_i$.
	Since $\sigma$ is a uniformly executable strategy, this means that $a_{\Kh_i}$ is executable on $\rho_h$.
	By the definition of $\M^\T$ in Definition \ref{def.inducedModel}, we know that there exists a maximal $\epsilon$-path $\rho'$ of $\T$ such that $\rho_h\rel{a_{\Kh_i\psi}}\rho'$, that is, $\rho_h\rho'$ is a path of $\T$ and $L(end(\rho_h),ini(\rho'))=a_{\Kh_i\psi}$.
	Note that such labels can only be added by the construction step (j) or (k).

	From the construction steps (j) and (k), we then have that $\Kh_i\psi,\neg\K_i\psi\in L(\rho_h)$.
	Since we also have that $\neg\Kh_i\chi\in L(\rho_h)$, by the construction step (k), there is a node $n_{(\neg\Kh_i\chi,\Kh_i\psi)}$ in $\T_{\phi_0}$ such that $L(n_{(\neg\Kh_i\chi,\Kh_i\psi)})=\{\K_i\psi,\neg\Kh_i\chi\}$ and $L(end(\rho_h),n_{(\neg\Kh_i\chi,\Kh_i\psi)})=a_{\Kh_i\psi}$.

	Since $\T$ is a complete subtree of $\T_{\phi_0}$, it follows that $n_{(\neg\Kh_i\chi,\Kh_i\psi)}$ is also a node in $\T$.
	There are two cases: $n_{(\neg\Kh_i\chi,\Kh_i\psi)}$ is blocked or not.

	If $n_{(\neg\Kh_i\chi,\Kh_i\psi)}$ is not blocked,	let $\rho$ be the maximal $\epsilon$-path in $\T$ that begins with the node $n_{(\neg\Kh_i\chi,\Kh_i\psi)}$. By the definition of $\M^\T$, we know that $(\rho_h,\rho)\in R_{a_{\Kh_i\psi}}$. This means that $[\rho_1]^i\cdots [\rho_h]^i [\rho]^i$ is a $\sigma$-execution. Moreover, due to $\neg\Kh_i\chi\in L(ini(\rho))$, it follows that $\neg\Kh_i\chi\in L(\rho)$. 

	If $n_{(\neg\Kh_i\chi,\Kh_i\psi)}$ is blocked by its some node $n'$, then $n'$ is an ancestor of $n$. So, we have that $n'$ is a node in $\T$. Let $\rho''$ be the maximal $\epsilon$-path that contains $n'$. Since block nodes are leaf nodes, it follows that $\rho''$ is not blocked. Thus, $\rho''\in W$. By the definition of $\M^\T$, we have that $(\rho_h,\rho'')\in R_{a_{\Kh_i\psi}}$.
	This means that $[\rho_1]^i\cdots [\rho_h]^i [\rho'']^i$ is a $\sigma$-execution. Moreover, due to $\neg\Kh_i\chi\in L(n)$, $L(n)=L(n')$ and $L(n')\subseteq  L(\rho)$, it follows that $\neg\Kh_i\chi\in L(\rho)$.
\end{proof}

\begin{lemma}[Truth lemma]\label{lemma.forComplete}
	If $\T$ is a complete and open subtree of $\T_{\phi_0}$, then for each $\phi\in\subpl{\phi_0}$, we have that $\M^\T,\rho\vDash\phi$ iff $\phi\in L(\rho)$.
\end{lemma}
\begin{proof}
	We prove it by induction on $\phi$. Since $L(\rho)$ is closed over $sub^+$-formulas, the atomic case and Boolean cases are straightforward. Next, we will focus on the cases of $\K_i\psi$ and $\Kh_i\psi$.

	For the case $\K_i\psi\in L(\rho)$, given $\rho'\in [\rho]^i$, we will show that $\M^\T,\rho'\vDash\psi$. By Proposition \ref{pro.equivalenceShareKAndKh}, we have that $\K_i\psi\in L(\rho')$.
	Moreover, 	since $\T$ is open and $\rho'$ is not blocked, this means that all formulas in $L(\rho')$ are marked as ``checked". Let $\rho'=n_1\cdots n_h$. From the construction, we know that there is some node $n_k$, where $1\leq k\leq h$, such that the first time in $\rho'$ the formula $\K_i\psi$ is marked as ``checked", and this only can be done by the construction step (d). Thus, we have that $\psi\in L(n_k)$, and then $\psi\in L(\rho')$. By IH, we have that $\M^\T,\rho'\vDash\psi$.

	For the case $\K_i\psi\not\in  L(\rho)$, since $L(\rho)$ is closed over $sub^+$-formulas, it follows that $\neg\K_i\psi\in L(\rho)$. Let $end(\rho)$ be the node $n$.
	Since $\rho$ is a non-blocked maximal $\epsilon$-path and $\T$ is open, this means that all formulas in $L(n)$ is marked as ``checked". From the construction, we know that the construction step (i) will be triggered on $n$. Thus, either there is an $i$-ancestor node $n'$ such that $\neg\psi\in L(n')$ or there is an $n$'s $i$-child node $n''$ such that $\neg\psi\in L(n'')$. Since $\T$ is a complete subtree of $\T_{\phi_0}$, it follows that no matter $n'$ or $n''$ will be included in $\T$.  Thus, there is $\rho'\in [\rho]^i$ such that $\neg\psi\in L(\rho')$. Since $\T$ is open, it follows that $\psi\not\in L(\rho')$. By IH, we have that $\M^\T,\rho\not\vDash\psi$. Hence, $\M^\T,\rho\not\vDash\K_i\psi$.

	For the case $\Kh_i\psi\in L(\rho)$, since $L(\rho)$ is closed over $sub^+$-formulas, it follows that either $\K_i\psi\in L(\rho)$ or $\neg\K_i\psi\in L(\rho)$.

	If $\K_i\psi\in L(\rho)$, by the proof of the $\K_i$-case above, we know that $\M^\T,\rho\vDash\K_i\psi$. By Proposition \ref{pro.validFormulas}, we have that $\M^\T,\rho\vDash\Kh_i\psi$.

	If $\neg\K_i\psi\in L(\rho)$, we will show that $\sigma$ is a good strategy for $\M^\T,\rho\vDash\Kh_i\psi$ where $\sigma$ is the function $\{[\rho]^i\mapsto a_{\Kh_i\psi}\}$.

	If $\rho'\in [\rho]^i $ and $(\rho',\rho'')\in R_{a_{\Kh_i\psi}}$, then either $\rho'\rel{a_{\Kh_i\psi}}\rho''$, or $\rho'\rel{a_{\Kh_i\psi}}\rho'''$  where $\rho'''$ is blocked by $\rho''$.
	For the first case, it is obvious that $\rho''\not \in [\rho]^i$. For the second case, we also have that $\rho''\not\in [\rho]^i$. The reason is that nodes can only be blocked in the construction step (k), so for the second case we have $\K_i\psi\in L(\rho''')$ and $L(\rho''')\subseteq L(\rho'')$.
	Due to $\neg\K_i\psi\in L(\rho)$ and Proposition \ref{pro.equivalenceShareKAndKh}, Therefore, for the second case we also have that $\rho''\not\in [\rho]^i$.
	Thus, we have shown that if $\rho'\in [\rho]^i $ and $(\rho',\rho'')\in R_{a_{\Kh_i\psi}}$, then $\rho''\not\in [\rho]^i$,
	which means that 
	$[\rho']^i[\rho'']^i$ is a complete $\sigma$-execution from $[\rho]^i$. 
	Hence, to show that $\M^\T,\rho\vDash\Kh_i\psi$, we only need to show that for each $\rho'\in [\rho]^i$,  (1) the action $a_{\Kh_i\psi}$ is executable at $\rho'$, (which  means that $\sigma$ is a uniform executable strategy,) and (2) $\M^\T,\rho''\vDash\K_i\psi$ for each $\rho''$ with $(\rho',\rho'')\in R_{a_{\Kh_i\psi}}$.

	For (1), due to $\Kh_i\psi,\neg\K_i\psi\in L(\rho)$ and Proposition \ref{pro.equivalenceShareKAndKh}, we know that $\Kh_i\psi,\neg\K_i\psi\in L(\rho')$. Thus, by the construction step (j), there is a node $n_{\Kh_i\psi}$ such that $L(n_{\Kh_i\psi})=\{\K_i\psi\}$ and $L(end(\rho'),n_{\Kh_i\psi})=a_{\Kh_i\psi}$.
	Since $\T$ is a complete subtree of $\T_{\phi_0}$, it follows that $n_{\Kh_i\psi}$ is a node in $\T$.
	Let $\rho'''$ be the maximal $\epsilon$-path that begins with $n_{\Kh_i\psi}$. We have that $\rho'\rel{a_{\Kh_i\psi}}\rho'''$, and then $(\rho',\rho''')\in R_{a_{\Kh_i\psi}}$. Thus,
	the action $a_{\Kh_i\psi}$ is executable at $\rho'$.

	For (2), if  $(\rho',\rho'')\in R_{a_{\Kh_i\psi}}$, it means that either  $L(end(\rho'),ini(\rho''))=a_{\Kh_i\psi}$ or $L(end(\rho'),n')=a_{\Kh_i\psi}$ where $n'$ is blocked by $\rho''$.
	From the construction, we know that only the steps (j) and (k) can add such labels. From these steps, we know that $\K_i\psi\in L(ini(\rho''))$ or $\K_i\psi\in L(n')$.  In either case, we have that $\K_i\psi\in L(\rho'')$. By the proof of the $\K_i$-case above, we know that $\M^\T,\rho''\vDash \K_i\psi$.

	For the case $\Kh_i\psi\not\in L(\rho)$, assume that $\M^\T,\rho\vDash\Kh_i\psi$. By the semantics, there is an $i$-strategy $\sigma$ such that all compete $\sigma$-executions from $[\rho]^i$ are finite and $\M^\T,\rho'\vDash\K_i\psi$ for all $[\rho']^i\in\CELeafN(\sigma,[\rho]^i)$.
	Due to $\Kh_i\psi\not\in L(\rho)$, we have that $\neg\Kh_i\psi\in L(\rho)$. By Proposition \ref{pro.negKhWillBeInheritedByActionTransitions}, we know that there is a complete $\sigma$-execution $[\rho]^i\cdots [\rho']^i$ such that $\neg\Kh_i\psi\in L(\rho')$.  By the construction step (g), we know that $\neg\K_i\psi\in L(\rho')$.
	Due to $\Kh_i\psi\in\subpl{\phi_0}$, it follows that $\K_i\psi\in\subpl{\phi_0}$.
	By proof of the $K_i$-case above, we have that $\M^\T,\rho'\vDash\neg\K_i\psi$. Contradiction!
	Thus, $\M^\T,\rho\not\vDash\Kh_i\psi$.
\end{proof}

The completeness below follows from Lemma \ref{lemma.forComplete} above.

\begin{theorem}[Completeness]\label{theo.complete}
	If $\T_{\phi_0}$ is not closed, 
	then $\phi_0$ is satisfiable.
\end{theorem}

\subsection{Complexity}
In this section, we will show that the procedure of Definition \ref{def.tableaux} runs in polynomial space.

\begin{definition}[Depth]
	The \emph{depth} of a formula, denoted by $\dep{\phi}$, is the depth of nesting of $\K$ or $\Kh$ operators, which is defined in Table \ref{tab.depth}. 
\end{definition}

\begin{table}[htbp]
	{
		$$\begin{array}{rll}
				\dep{\bot}          & = & 0                            \\
				\dep{p}             & = & 0                            \\
				\dep{\neg\psi}      & = & \dep{\psi}                   \\
				\dep{\psi\land\chi} & = & max\{\dep{\psi},\dep{\chi}\} \\
				\dep{\K_i\psi}      & = & \dep{\psi}+1                 \\
				\dep{\Kh_i\psi}     & = & \dep{\psi}+2
			\end{array}$$

	}
	\caption{Depth of formulas}\label{tab.depth}
\end{table}


The reason that $\dep{\Kh_i\psi}>\dep{\K_i\psi}$ is that in the construction steps (j) and (k) we need to label $n$'s child node with the formula $\K_i\psi$ for $\Kh_i\psi\in L(n)$.

\begin{proposition}\label{pro.heightOfTableauTree}
	The height of the tableau tree $\T_{\phi_0}$ is bounded by a polynomial function of the size of the set $\subpl{\phi_0}$.
\end{proposition}
\begin{proof}
	Let the size of $\subpl{\phi_0}$ be $m$. In the construction of the tableau tree $\T_{\phi_0}$ in Definition \ref{def.tableaux}, each step of (a)-(k) might add the height of the tree with 1 degree.
	The steps (a)-(g) can be consecutively executed at most $m$ times to get a $sub^+$-formula closure.

	Now consider a $\T_{\phi_0}$ path that starts from the root node, and we contract it by seeing the consecutive nodes whose edges are labeled $\epsilon$ as one single node. Let $n$ be a node in this contracted path $b$.

	If the node $n$'s child node is added by the step (j), then the greatest depth of formulas labeled on the child node is strictly less than the greatest depth of formulas labeled on $n$.

	If the node $n$'s child node is added by the step (i), although the greatest depth of formulas labeled on the child node might be the same as the greatest depth of formulas labeled on $n$, but, such descendant nodes with the same greatest depth with $n$ can be consecutively added by the step (i) at most $m$ times to achieve a node whose $K_i$ ancestor has the same labeled formulas. After that, if a descendant node $n'$ is added by the step (i) again, then it must be a $\K_j$ descendant where $i\neq j$. 
	Thus, the greatest depth will be strictly shrunk.
	If the descendant node $n'$ is added by the step (k), the greatest depth might still keep the same with $n$. However, the step (k) can be executed on one path at most $m^2$ times before it adds a blocked node.

	Hence,  there are at most $m^4$ consecutive nodes in $b$ that have the same greatest depth before the greatest depth becomes strictly less.
	Therefore, the length of the contracted path $b$ is at most $m^5$, and the length of the original path is at most $m^6$. It follows that the height of $\T_{\phi_0}$ is bounded by $m^6$.
\end{proof}

\begin{lemma}\label{lemma.upperBound}
	There is an algorithm that runs in polynomial space for deciding whether $\T_{\phi_0}$ is closed.
\end{lemma}
\begin{proof}
	Let the size of $\subpl{\phi_0}$ be $m$.
	In the construction step (i), at most $m$ successors are added. In the step (j), at most $m^2$ successors are added, and in the step (k) at most $m^3$ successors are added. This means that each node has at most $m+m^2+m^3$ child nodes.  Therefore, the tableau tree $\T_{\phi_0}$ is an  $m+m^2+m^3$-ary tree. By Proposition \ref{pro.heightOfTableauTree}, we know that the height of the tree is bounded by $m^6$.

	We can mark the node to check whether the tree is closed.
	Note that how a node is marked can be completely determined by its label and how its successors are marked. Once we have determined how  a node is marked, we never have to consider the subtree below that node again.
	Thus, a depth-first search of the tree that runs in polynomial space can decide whether the tree is closed \cite{halpern_guide_1992}.
\end{proof}

Since the satisfiability problem of epistemic logic with no less than 2 agents is \pspace-hard (see \cite{halpern_guide_1992}), and it is a fragment of $\ELKh_n$, it follows that the satisfiability problem of $\ELKh_n$ is also \pspace-hard. Together with Lemma \ref{lemma.upperBound},
we have the following result.

\begin{theorem}
	The problem of the satisfiability of $\ELKh_n$ formulas is \pspace -complete.
\end{theorem}

\section{Conclusion}
This paper presented a tableau procedure for the multi-agent version of the logic of strategically \emph{knowing how}.
The tableau method presented in this paper is developed from the tableau method for epistemic logic in \cite{halpern_guide_1992} and the tableau method for \emph{knowing how} logic via simple plans \cite{Li21tableau}.
This paper showed the soundness, the completeness, and the complexity of this tableau procedure.
Since the procedure runs in polynomial space, it follows that the satisfiability problem of the logic of strategically \emph{knowing how} can be decided in \pspace.
Moreover, since the \emph{knowing how} logic based on PDL-style knowledge-based plans over finite models in \cite{liWang2021AIJ} is the same as the logic of strategically \emph{knowing how}, it means that the satisfiability problem of that logic also can be decided in \pspace.

\paragraph{Acknowledgement}

The author is grateful to the anonymous reviewers of TARK-2023. Their detailed comments helped the author to improve the presentation of the paper. The author also thanks the support from the National Social Science Foundation of China (Grant No. 18CZX062).

\bibliographystyle{eptcs}
\bibliography{generic}
\end{document}